\documentclass[11pt]{article}
\usepackage{graphicx} 
\usepackage[english]{babel}
\usepackage[utf8]{inputenc}
\usepackage[a4paper, width=150mm,top=25mm,bottom=25mm]{geometry}
\usepackage{mathpazo}
\usepackage{setspace}

\usepackage{graphicx}
\usepackage{appendix}
\usepackage{epigraph}
\usepackage{longtable}
\usepackage{array}
\usepackage{caption}
\usepackage{subcaption}
\usepackage{hyphenat}
\usepackage{caption}
\usepackage{multirow}
\usepackage{longtable}
\usepackage{makecell}
\usepackage{bbding}
\usepackage{booktabs}
\usepackage{dcolumn}
\usepackage{algorithm}
\usepackage{bm}
\usepackage{amsmath}
\usepackage{mathrsfs}
\usepackage{tikz}
\usepackage{tabularx}
\usetikzlibrary{calc}
\usepackage{tabularray}
\usepackage{multirow,array}
\usepackage{comment}
\usepackage{quiver}

\usepackage{amsthm}
\usepackage{amssymb}
\usepackage{bm}
\usepackage[bbgreekl]{mathbbol}
\usepackage{xspace}
\usepackage{afterpage}

\usepackage{booktabs} 
\usepackage{float}
\usepackage{multirow}
\usepackage{mathtools}
\usepackage{natbib}
\usepackage[flushleft]{threeparttable}
\newtheorem{definition}{Definition}
\newtheorem{theorem}{Theorem}
\newtheorem{corollary}{Corollary}

\usepackage{authblk}

\usepackage{tikz}
\foreach \a in {q,w,e,r,t,y,u,i,o,p,a,s,d,f,g,h,j,k,l,z,x,c,v,b,n,m,Q,W,E,R,T,Y,U,I,O,P,A,S,D,F,G,H,J,K,L,Z,X,C,V,B,N,M}{
  \expandafter\xdef\csname v\a\endcsname {
		{\noexpand\mathbf{\a}}
	}
}

\newtheorem{remark}{Remark}

\usepackage[T1]{fontenc}
\usepackage{lmodern}
\title{Granger Causality in Expectiles: an M-vine copula test}
\author[1,2]{Roberto Fuentes-Martínez}
\author[1]{Irene Crimaldi}
\date{}
\affil[1]{\small IMT School for Advanced Studies Lucca}
\affil[2]{\small Universidad de Alicante}

\begin{document}

\maketitle

{\bf Abstract.} {A model-free measure of Granger causality in expectiles is proposed, generalizing the traditional mean-based measure to arbitrary positions of the conditional distribution. Expectiles are the only law-invariant risk measures that are both coherent and elicitable, making them particularly well-suited for studying distributional Granger causality where risk quantification and forecast evaluation are both relevant. Based on this measure, a test is developed using M-vine copula models that accounts for multivariate Granger causality with $d+1$ series under non-linear and non-Gaussian dependence, without imposing parametric assumptions on the joint distribution. Strong consistency of the test statistic is established under some regularity conditions. In finite samples, simulations show accurate size control and power increasing with sample size. A key advantage is the joint testing capability: causal relationships invisible to pairwise tests can be detected, as demonstrated both theoretically and empirically. Two applications to international stock market indices at the global and Asian regional level illustrate the practical relevance of the proposed framework.}

\section{Introduction}

The classical Granger causality framework (Granger, 1969) assesses whether lagged values of one stochastic process provide statistically significant information for forecasting another, conditional on the past of the latter. While widely applied, this approach is inherently limited to linear predictive structures in the conditional mean, rendering it inadequate in the presence of non-linear, asymmetric, or tail-dependent relationships that frequently characterize financial, macroeconomic, and environmental time series.
\\

To overcome these limitations, recent methodological advances have extended Granger causality to notions that go beyond the conditional mean, offering a distributional perspective on causal dependence. Most of this literature has focused on the notion of Granger causality in quantiles. For instance, the parametric approach from \cite{troster2018}, the non-parametric approaches from \cite{jeong2012} and \cite{balcilar2016}, and the copula-based framework from \cite{lee2014}, among others. Later, \cite{song2021} introduced model-free measures of Granger causality in quantiles, extending the approach of \cite{song2018} to different parts of the distribution. Based on this measure, \cite{jang2023} proposed a test for Granger causality in quantiles using vine copulas. More recently, \cite{bouezmarni2024} introduced the first parametric test for Granger non-causality in expectiles based on expectiles regressions. Despite this progress, no model-free or copula-based tests for Granger causality in expectiles have been proposed, particularly one that can accommodate multivariate causality under non-linear and non-Gaussian dependence.\\

Shifting the focus from Granger causality in quantiles to expectiles is motivated from the theory of risk measures. Quantiles have been widely used in financial and risk management literature under the name of Value-at-Risk, nonetheless, quantiles are not a coherent risk measure in the sense of \cite{artzner1999} as they lack sub-additivity. Expectiles were introduced by \cite{newey1987} as the minimizers of an asymmetric quadratic loss and they constitute the only class of law-invariant risk measures that are both coherent and elicitable \citep{bellini2015,ziegel2016}. These two properties make expectiles particularly well-suited for distributional Granger causality analysis in contexts where risk quantification and forecast evaluation are both relevant.\\

Moreover, the incorporation of copula theory provides a flexible and model-agnostic tool for capturing non-linear and asymmetric dependence structures between variables. Copulas allow the joint distribution of time series to be decomposed into their marginal distributions and a dependence function, facilitating the explicit modeling of tail dependence and other complex interactions that traditional linear models fail to capture. For a comprehensive treatment of copula theory we refer to \cite{nelsen2007}, \cite{joe2014}, and \cite{durante2015}. More recently, vine copulas, that are hierarchical constructions that decompose high-dimensional copulas into a cascade of bivariate building blocks, have been employed to model both cross-sectional and temporal dependence in stationary time series. Some examples of these structures are the M-vines \citep{beare2015}, the COPAR-vines \citep{brechmannczado2015}, and the D-vines \citep{smith2015}. Building on their work, \cite{nagler2022} generalized these models introducing the S-vines, a broader class of all the vine structures that can represent stationary time series.\\ 

Integrating copula-based dependence modeling with expectile-based Granger causality constitutes a framework that jointly addresses distributional asymmetry, tail sensitivity, and non-linear dependence, yielding a more comprehensive characterization of causal mechanisms in complex stochastic systems. Such a framework has broad applicability, particularly in domains where extreme co-movements and heterogeneous dependence play a crucial role, such as risk management, financial contagion, and climate dynamics.
\\

The contribution of the present paper is threefold. First,  we define a model-free measure of Granger causality in expectiles that generalizes the mean-focused measure of \cite{song2018} to other positions of the conditional distribution. This measure inherits the desirable properties of its predecessors: being non-negative and equal to zero if and only if there is no Granger causality in the $\tau$-th expectile. Moreover, it reduces to the mean-based measure of \cite{song2018} when $\tau = 1/2$.  Second, we propose a test for Granger causality in expectiles based on the M-vine copula model of \cite{beare2015}, extending the bivariate vine copula testing framework of \cite{jang2022} and \cite{fuentesmartinez2025}
to the expectile domain and, crucially, to the multivariate setting with $d+1$ series. Indeed, as we will illustrate both in simulations and with an empirical application, there are scenarios in which pairwise dependencies are present, but too small to be detected by pairwise tests, while the joint effect is sizable enough for the proposed multivariate test to catch it. 
In addition, we will also construct 
a theoretical case in which pairwise Granger causality in the mean is entirely absent, while joint Granger causality in the mean is present and, with simulations, we show that our test is able to detect it. Theoretical results supporting
the consistency of the estimators are provided and, through a simulation study, we also demonstrate that the proposed test has good finite-sample size control and it exhibits a power that increases with sample size $T$ until reaching excellent values already for $T\geq 200$. Moreover, for $\tau=1/2$ (the case of Granger causality in the mean) it has been possible to compare the proposed method with the classical linear Granger causality in the mean F-test and the comparison has revealed a comparable size, but a strongly superior power in the presence of non-linear dependencies.  
Third, we provide two empirical applications of the proposed test to international stock markets: in particular, the joint testing reveals causal relationships that are invisible to pairwise tests. 
Indeed, in one of them,  pairwise Granger causality is not strong enough to be significant for each series, but the joint test clearly rejects the null hypothesis of no Granger causality for expectiles located in the left tail of the distribution. 
\\

The remainder of this paper is organized as follows. Section~\ref{gcexp} defines Granger causality in expectiles and proposes the model-free measure. Section~\ref{mvtest} describes the M-vine copula test procedure. Section~\ref{simstudy} presents the simulation study. Section~\ref{empapp} applies the test to global and Asian stock market data. Section~\ref{conc} concludes. Theoretical results on the consistency of the estimators are gathered in the Appendix.

\section{Granger causality in expectiles}\label{gcexp}

Let $(X_t, \vZ_t)_{t\in\mathbb{Z}}$ be a (strictly) stationary stochastic process,  with $X_t$ taking values in ${\mathbb R}$ and $\vZ_t=(Z_{1,t},\dots,Z_{d,t})$ taking values in ${\mathbb R}^{d}$,  that is 
$$
[(X_t, \vZ_t),\dots, (X_{t+h}, \vZ_{t+h})] \stackrel{D}=
[(X_s, \vZ_s),\dots, (X_{s+h}, \vZ_{s+h})]\qquad\forall s,\,t\in\mathbb{Z},\, h\in\mathbb{N}
$$ 
where the symbol $\stackrel{D}=$ means equality in distribution. Suppose further that $(X_t, \vZ_t)_{t\in\mathbb{Z}}$ is a Markov process of order $k\in \mathbb{N}\setminus\{0\}$ (more briefly, a $k$-Markov process), that is 
$$
[X_{t},\vZ_t] \,|\, [(X_s,\vZ_s):\,s\leq t-1]\stackrel{D}= [X_{t},\vZ_t] \,|\, [(X_{t-1},\vZ_{t-1}),\dots, (X_{t-k},\vZ_{t-k})]
\qquad\forall t\,.
$$ 
Set $X=(X_t)_{t\in\mathbb{Z}}$ and $\vZ=(\vZ_t)_{t\in\mathbb{Z}}$. Moreover, 
denote by ${\mathcal I}_{X}(s)$ the $\sigma-$field $\sigma(X_u: u\leq s)$, that is the information brought by $X$ until time-step $s$ and by ${\mathcal I}_{X, \vZ}(s)$ the $\sigma-$field $\sigma((X_u, \vZ_u): u\leq s)$, that is the information brought by $X$ and $\vZ$ until time-step $s$. Assuming $X_t$ and $\vZ_t$ to be square integrable, we say that $\vZ$ causes $X$ (in the mean) 
in the sense of Granger, if  
\begin{equation*}
\sigma^2(X_t|\, {\mathcal I}_{X,\vZ}(t-k) \,) < \sigma^2(X_t|\, {\mathcal I}_{X}(t-k) \,)
\end{equation*}
where 
\begin{equation*}
\begin{split}
&\sigma^2(X_t|\, {\mathcal I}_{X}(t-k) \,)=\mathbb{E}\left[(X_t-\mathbb{E}[X_t|\, {\mathcal I}_{X}(t-k)])^2\right]
\quad\mbox{and}\\
&\sigma^2(X_t|\, {\mathcal I}_{X,\vZ}(t-k) \,)=\mathbb{E}\left[(X_t-\mathbb{E}[X_t|\, {\mathcal I}_{X,\vZ}(t-k) ])^2\right]
\end{split}
\end{equation*} 
are the mean squared errors of the optimal prediction of $X_t$ 
given the information ${\mathcal I}_{X}(t-k)$ and given the information ${\mathcal I}_{X,\vZ}(t-k)$, respectively. \cite{song2018} proposed a model-free measure of Granger causality in the mean for the bivariate case, which can be easily adapted into our multivariate setting as

\begin{equation*}
GC^{mean}(\vZ\rightarrow X)=\log 
\left[
\frac
{\sigma^2(X_t|\, {\mathcal I}_{X}(t-k) \,)}
{\sigma^2(X_t|\, {\mathcal I}_{X,\vZ}(t-k) \,)}
\right].
\end{equation*}

While most of the literature has focused on the notion of Granger causality in the mean, this can also be extended for other parts of the distribution. For instance, several tests have been proposed for testing Granger causality in quantiles \citep{jeong2012,lee2014,balcilar2016,troster2018,jang2023}. Moreover, \cite{song2021} also propose a model-free measure of Granger causality in quantiles. Given that quantiles are often used as risk measures in financial risk management, testing for Granger causality in quantiles can be understood as testing if the past information of previous extreme events from other risk factors can improve the prediction of future extreme events of a given risk factor of interest. However, quantiles are not a coherent risk measure as they are not sub-additive. Hence, in order to properly measure Granger causality in risk, one could use other risk measures, such as expectiles. Expectiles are the only law invariant coherent and elicitable risk measure \citep{bellini2015,ziegel2016}. They are defined as the minimizers of an asymmetric quadratic loss: formally, given $\tau\in (0,1)$, the {\em $\tau$-th expectile of $X$} is defined as 
$$
\mu_{\tau} (X)=\arg\min_{m\in\mathbb{R}} \mathbb{E}\left[\,\mathcal{R}_\tau(X-m)\,\right]
$$
where 
\begin{equation*}
\begin{split}
\mathcal{R}_\tau(X-m)&= |\,\tau - \mathbb{I}_{\{(X-m)<0\}} \,|\, (X-m)^2 \\
&=
\left(\tau\,\mathbb{I}_{\{(X-m)\geq  0\}}+(1-\tau)\,\mathbb{I}_{\{(X-m)<0\}}\right)\,(X-m)^2\\
&=
\tau\, (X-m)_+^2 + (1-\tau)\,(X-m)_-^2
\end{split}
\end{equation*}
with $x_+=\max\{x,0\}$ and $x_-=\max\{-x,0\}$.\\

Using the definition of expectiles, we can provide a definition of Granger causality in expectiles.

\begin{definition}{\textbf{(Granger causality in the $\tau$-th expectile)}}
Assuming $X_t$ and $\vZ_t$ to be square integrable and given $\tau\in (0,1)$, we say that $\vZ$ Granger-causes $X$ through its $\tau$-th expectile, if  
\begin{equation*}
\mathbb{E}\left[ \mathcal{R}_\tau(X_t-\mu_\tau(X_t|\, {\mathcal I}_{X,\vZ}(t-k))\right ] < \mathbb{E}\left[ \mathcal{R}_\tau(X_t-\mu_\tau(X_t|\, {\mathcal I}_{X}(t-k))\right ]\,,
\end{equation*}
where 
\begin{equation*}
\begin{split}
\mu_\tau(X|\, {\mathcal I})&=\arg\min_{m\in\mathbb{R}} \mathbb{E}\left[\mathcal{R}_\tau(X-m)\,|\,{\mathcal I}\right]
\\
&=
\arg\min_{m\in\mathbb{R}} \left( \tau\,\mathbb{E}[(X-m)_+^2\,|\,\mathcal{I}]+(1-\tau)\,\mathbb{E}[(X-m)_-^2\,|\,{\mathcal I}]\right).
\end{split}
\end{equation*}
\end{definition}

In order to measure the degree of causality in the $\tau$-th expectile from $\vZ$ to $X$, one could use the difference between these two loss functions. This allows us to propose a model-free measure of Granger causality in expectiles.

\begin{definition}{\textbf{($\tau$-th expectile Granger causality measure)}} Assuming $X_t$ and $\vZ_t$ to be square integrable and given $\tau\in (0,1)$, we define the $\tau$-th expectile Granger causality measure from $\vZ$ to $X$ as
$$
GC_{\tau}(\vZ \rightarrow X)=\log 
\left[
\frac
{\mathbb{E}\left[ \mathcal{R}_\tau(X_t-\mu_\tau(X_t|\, {\mathcal I}_{X}(t-k)))\right ]}
{\mathbb{E}\left[ \mathcal{R}_\tau(X_t-\mu_\tau(X_t|\, {\mathcal I}_{X,\vZ}(t-k)))\right ]}
\right]\,,
$$
provided the two mean values are strictly positive. 
\end{definition}

The expectile causality measure $GC_{\tau}(\vZ \rightarrow X)$ is aligned with the definition of measure functions of dependence and feedback between time series in \cite{geweke1982}, as well as with the ones introduced in \cite{song2018,song2021}. In fact, it has the relevant properties that these measures include, such as being non-negative and canceling only when there is no Granger causality.

	\begin{remark}\label{remgen}\rm 
	For $\tau=1/2$ the $GC_{\tau}(\vZ \rightarrow X)$ is equal to $GC^{mean}(\vZ \rightarrow X)$ by \cite{song2018}.	This follows from the fact that when  $\tau=\frac{1}{2}$, the asymmetric quadratic loss $\mathcal{R}_\tau$ is equal to the mean squared error, and given that the minimizer of the mean squared error is the mean, it implies that $\mu_{1/2}(X)=\mathbb{E}[X]$. Hence, $GC_{1/2}(\vZ \rightarrow X)=GC^{mean}(\vZ\rightarrow X)$.
	\end{remark}

\section{M-vine test for Granger causality in expectiles}\label{mvtest}

Let $(X,Z_1,\dots,Z_d)=(X_t,Z_{1,t},\dots,Z_{d,t})_t$ be a $k-$Markov square-integrable (strictly) stationary stochastic process and let $\{(x_t,z_{1,t},\dots,z_{d,t})\;:\;t=1,\dots,T\}$ be a sample of it. For simplicity, we assume that $k=1$, but everything can be naturally extended to any order $k\in\mathbb{N}\setminus\{0\}$. \\


  
\noindent The following test is an extension of the one introduced in~\cite{fuentesmartinez2025} for Granger causality in the mean. \\
\indent We start by describing part A (computation of the value of the test statistic). We use the sample version  $GC_{\tau}(\vZ \rightarrow X)$ of the above proposed measure, that is 

\begin{equation*}
\begin{split}
\widehat{GC}_{\tau}(\vZ \rightarrow X)&=
\log\left[ 
\frac{(T-T_0+1)^{-1}\sum_{t=T_0}^{T}\mathcal{R}_\tau(x_t-\widehat{\mu}_\tau(X_t|\, X_{t-1}=x_{t-1})}
{(T-T_0+1)^{-1}\sum_{t=T_0}^{T}\mathcal{R}_\tau(x_t-\widehat{\mu}_\tau(X_t|\, X_{t-1}=x_{t-1},\,\vZ_{t-1}=\vz_{t-1} )}
\right]
\\
&=\log\left[ 
\frac{\sum_{t=T_0}^{T}\mathcal{R}_\tau(x_t-\widehat{\mu}_\tau(X_t|\, X_{t-1}=x_{t-1})}
{\sum_{t=T_0}^{T}\mathcal{R}_\tau(x_t-\widehat{\mu}_\tau(X_t|\, X_{t-1}=x_{t-1},\, \vZ_{t-1}=\vz_{t-1})}
\right]
\end{split}
\end{equation*}

where the expectiles $\widehat{\mu}_\tau(X_t|\, X_{t-1}=x_{t-1})$ and $\widehat{\mu}_\tau(X_t|\, X_{t-1}=x_{t-1},\,\vZ_{t-1}=\vz_{t-1})$ are computed fitting an M-vine copula model to the observed sample and using it to generate i.i.d. observations of $X_t$ given $X_{t-1}=x_{t-1}$ and of $X_t$ given $X_{t-1}=x_{t-1}$ and $\vZ_{t-1}=\vz_{t-1}$, and using these generated samples to compute the corresponding expectiles empirically. More precisely, we proceed as follows:

\begin{itemize}
\item[Step A1)] we fit an M-vine copula model and estimate its respective parameters for the observations from the series $X$, that is, for  $M_{x}=(x_t)_{t=1\dots,T}$, by means of the selection and estimation procedure introduced and studied in \cite{nagler2022};\\
\item[Step A2)] using the model obtained from Step A1), for each $t=T_0,\dots, T$,  we 
generate $N$ i.i.d. predictions $\{\tilde{x}^{M_{x}}_{i,t}:\, i=1,\dots,N\}$ of $X_t$ given $X_{t-1}=x_{t-1}$, so that we can compute the empirical conditional $\tau$-th expectile of $X_t$ given $X_{t-1}=x_{t-1}$ by the sample $\{\tilde{x}^{M_{x}}_{i,t}\}_{i=1,\dots,N}$ , that is we set 
\begin{equation*}
\widehat{\mu}_\tau(X_t|\, X_{t-1}=x_{t-1})=
\widehat{\mu}_{\tau,N}(X_t|\, X_{t-1}=x_{t-1})
= \arg\min_m\sum_{i=1}^N \mathcal{R}_\tau(\tilde{x}^{M_{x}}_{i,t}-m)\,;
\end{equation*}
\item[Step A3)] repeat Step A1) and Step A2) for the observations from the series $(X,\vZ)$, 
that is for  $M_{x\vz}=(x_t,\vz_t)_{t=1\dots,T}$, in order to obtain 
\begin{equation*}
\begin{split}
\widehat{\mu}_\tau(X_t|\, X_{t-1}=x_{t-1},\,\vZ_{t-1}=\vz_{t-1})&= 
\widehat{\mu}_{\tau,N}(X_t|\, X_{t-1}=x_{t-1},\,\vZ_{t-1}=\vz_{t-1})\\
&=
\arg\min_m\sum_{i=1}^N \mathcal{R}_\tau(\tilde{x}^{M_{x\vz}}_{i,t}-m)\,.
\end{split}
\end{equation*}
\end{itemize}

\begin{remark} \rm Note that Theorem~\ref{th-strong-consistency} assures that the empirical quantity 
$\widehat{\mu}_\tau(X_t|\, X_{t-1}=x_{t-1})$ converges a.s. (for $N\to+\infty$) toward the ``pseudo-true'' conditional 
$\tau$-expectile of $X_t$ given $X_{t-1}=x_{t-1}$, that is 
toward the conditional $\tau$-expectile of $X_t$ given $X_{t-1}=x_{t-1}$  
with respect to the model obtained from Step A1. Indeed, it is enough to apply 
Theorem~\ref{th-strong-consistency} taking $X$ with distribution equal to the conditional distribution of $X_t$ given $X_{t-1}=x_{t-1}$ determined by the model obtained from step A1), so that $\{\tilde{x}^{M_{x}}_{i,t}: i=1,\dots,N\}$ is a realization of the random variables $X_1,\dots,X_N$ in the theorem and 
$\widehat\mu_\tau(X_t|X_{t-1}=x_{t-1})$ coincides with 
$\hat\mu_{\tau,N}(\tilde{x}^{M_{x}}_{1,t}, \dots,\tilde{x}^{M_{x}}_{N,t})$ with $\hat\mu_{\tau,N}(x_1,\dots,x_N)$ defined in the theorem. 
Similar arguments hold true for $\widehat{\mu}_\tau(X_t|\, X_{t-1}=x_{t-1},\, \vZ_{t-1}=\vz_{t-1})$. Hence, assuming that the fitted models of step A1 are the true ones, by Corollary~\ref{cor-1} and the continuity of the $\log(\cdot)$ on $(0,+\infty)$, we have that, for each fixed $T_0$ and $T$,
\begin{equation*}
\begin{split}
&\log\left[ 
\frac{\sum_{t=T_0}^{T}\mathcal{R}_\tau(X_t-\widehat{\mu}_{\tau,N}(X_t|\, X_{t-1})}
{\sum_{t=T_0}^{T}\mathcal{R}_\tau(X_t-\widehat{\mu}_{\tau,N}(X_t|\, X_{t-1},\, \vZ_{t-1})}
\right]
 \stackrel{a.s.}\longrightarrow \\
& \log\left[ 
\frac{\sum_{t=T_0}^{T}\mathcal{R}_\tau(X_t-\mu_\tau(X_t|\, X_{t-1})}
{\sum_{t=T_0}^{T}\mathcal{R}_\tau(X_t-\mu_\tau(X_t|\, X_{t-1},\, \vZ_{t-1})}
\right]\quad\mbox{for } N\to +\infty\,.
\end{split}
\end{equation*}
Moreover, assuming that the fitted models are also ergodic, by Corollary~\ref{cor-2} and again the continuity of the $\log(\cdot)$ on $(0,+\infty)$, 
the sample version $\widehat{GC}_\tau(\vZ\to X)$ results (taking first $T\to +\infty$ and then $N\to +\infty$) a strongly consistent estimator 
of $GC_\tau(\vZ\to X)$.
\end{remark}

Roughly speaking, $\widehat{GC}_{\tau}(\vZ \rightarrow X)$, measures the log-difference between the mean $\tau$-expectile loss related to the model 
with an M-vine copula structure fitted only on the information set from $X$,  
and the mean $\tau$-th expectile loss computed with the M-vine copula model fitted to the entire sample 
from both $X$ and $\vZ$, that is  using the information set of all series. 
In line with the definition of Granger causality in $\tau$-th expectiles, if this estimated quantity is significantly 
higher than zero, we reject the null hypothesis of no Granger causality (in the $\tau$-th expectile) from $\vZ$ to $X$: 
indeed, under the null hypothesis, the measure $GC_{\tau}(\vZ \rightarrow X)$ is zero, and so, the higher is the value of the 
statistics $\widehat{GC}_{\tau}(\vZ \rightarrow X)$, the more statistically significant is the evidence of the presence of 
Granger causality (in the $\tau$-th expectile) running from $\vZ$ to $X$. Regarding $T_0\geq k+1$, 
we can choose it by taking into account the goodness of fit of the models to the data. 
\\

\indent In order to test if $\widehat{GC}_{\tau}(\vZ \rightarrow X)$ is statistically greater than zero, we rely, as in \cite{jang2022}, on a method, which takes 
inspiration from \cite{paparoditis2000} and also used in \cite{fuentesmartinez2025}.  
This method relies on the M-vine copula model fitted on the entire sample $M_{x\vz}=\{(x_t,\vz_t):t=1,\dots,T\}$ in order to 
generate independent samples under the null hypothesis of no Granger causality (in the $\tau$-th expectile) that have the same dependence structure as the original sample, and use them for computing the $p-$value for the test. Previous literature has dealt with these methods in order to simulate exclusively bivariate series, hence, we propose an extension of it to the multivariate case in which we allow for $d+1$ series. Given the presence of more than two series, the methodology is not as straightforward as in the original case in which all the required copulas are present in the first tree of the M-vine structure. Therefore, the multivariate case requires to proceed sequentially using the first $d$ trees of the corresponding vine. More precisely, the part B (computation of the $p$-value) of the proposed procedure works as follows:
\begin{itemize}  
\item[Step B1)]  from the first tree of the obtained M-vine copula structure, extract, for each $t=1,\dots, T$, both the copula between
 $X_t$ and $X_{t+1}$, say $c_{X_t,X_{t+1}}$, related to the conditional distribution of $X_{t+1}$ given $X_t$, 
 and the copula between $X_t$ and $Z_{1,t}$, say $c_{X_t,Z_{1,t}}$, related to the conditional distribution of $Z_{1,t}$ given $X_t$;
\item[Step B2)] using the estimated marginal distribution $\widehat{F}_X$, generate $x_t^0$, and, conditional on this value, 
 draw $x_{t+1}^0$ from $c_{X_t,X_{t+1}}$; 
\item[Step B3)] using also the estimated marginal distribution $\widehat{F}_{Z_1}$,  
  conditional on $x_t^0$, draw $z_{1,t}^0$ from $c_{X_t,Z_{1,t}}$; 

\item[Step B4)] from Tree $i$, for $i = 2, \dots, d$, extract the conditional copula 
$c_{X_t, Z_{i,t};\, Z_{1,t},\dots,Z_{i-1,t}}$; then, using the estimated marginal 
distributions $\widehat{F}_{Z_2}, \dots, \widehat{F}_{Z_d}$, draw sequentially 
$z_{i,t}^0$ from the conditional distribution of $Z_{i,t}$ given 
$x_t^0, z_{1,t}^0, \dots, z_{i-1,t}^0$, for $i = 1, \dots, d$, yielding the full 
simulated vector $\mathbf{z}_t^0 = (z_{1,t}^0, \dots, z_{d,t}^0)$;

\item[Step B5)] using this generated sample $\{(x_t^0,\vz_t^0):t=1,\dots,T\}$, 
compute the quantity $\widehat{GC}^{0}_\tau(\vZ\rightarrow X)$, that gives a simulated value of the test statistic under the null hypothesis; 
  
\item[Step B6)] repeat the above steps $B$ times, so that we get $B$ independent simulated values under the null hypothesis:   
$\widehat{GC}^{0}_{\tau\;j}(\vZ\rightarrow X)$ for $j=1,\dots,B$;
\item[Step B7)] compute the $p$-value for the test by the empirical mean
\begin{equation*}
p=\frac{1}{B}\sum_{j=1}^B \mathbb{1}(\widehat{GC}^{0}_{\tau\;j}(\vZ\rightarrow X) \geq \widehat{GC}_\tau(\vZ\rightarrow X)).
\end{equation*}
\end{itemize}
Therefore, we reject the null hypothesis of no Granger causality (in the $\tau$-th expectile) when $p<\alpha$, where $\alpha$ is a given significance level.

\section{Simulation Study}\label{simstudy}

In order to analyze the finite sample properties of the proposed M-vine test for Granger causality in expectiles in terms of size and power, we perform a simulation study based on two size assessment models and four power assessment models. For the selected data generating processes (DGP's) we set $d=2$, i.e., we work with $d+1=3$ series. In order to ease both notation and exposition of the DGP's, we will denote the considered multivariate time series by $(X_t,Y_t,Z_t)_t$ and we will study the Granger causality in expectiles from $(Y,Z)\rightarrow X$. We simulate $S=500$ Monte Carlo replications for each model with sample sizes $T\in\{100,200,500\}$ and compute the empirical size and power of our test using a predefined significance level of $\alpha=0.05$. Furthermore we set $\tau\in\{0.1,0.5,0.9\}$, in order to study the statistical properties of our test both in the mean and at the tails of the distribution. In terms of the specification of our test, for Part A we perform $N=200$ predictions for each series and set $T_0=\frac{T}{2}$, whilst for part B, we work with $B=200$ generated samples under the null hypothesis. We begin by introducing the two size assessment models in our simulation study. 
\\

\textbf{\underline{Size assessment models}}\\
\begin{flalign*}
& \begin{aligned}[t]
\makebox[2em][l]{\textbf{S1}}\; X_t &= 0.5 X_{t-1} + \eta_{x,t}, \qquad Y_t = 0.5 Y_{t-1} + \eta_{y,t}, \qquad Z_t = 0.5 Z_{t-1} + \eta_{z,t},
\end{aligned} &&
\end{flalign*}
where $(\eta_{x,t})_t, (\eta_{y,t})_t, (\eta_{z,t})_t$ are three independent white Gaussian noises.

\medskip
\begin{flalign*}
& \begin{aligned}[t]
\makebox[2em][l]{\textbf{S2}}\; X_t &= 0.05 + \varepsilon_{x,t}, \qquad Y_t = 0.05 + \varepsilon_{y,t}, \qquad Z_t = 0.05 + \varepsilon_{z,t}, \\
\varepsilon_{i,t} &= \sigma_{i,t}\,z_{i,t}, \\
\sigma_{i,t}^2 &= 0.01 + 0.08\,\varepsilon_{i,t-1}^2 + 0.87\,\sigma_{i,t-1}^2, \qquad i \in \{x,y,z\},
\end{aligned} &&
\end{flalign*}
where $(z_{x,t})_t, (z_{y,t})_t, (z_{z,t})_t$ are three independent sequences of i.i.d. random variables distributed as the standardized skewed Student's t distribution~$\text{sstd}(0,1,\nu=5,\xi=-1.5)$.\\

Given that these models are employed to analyze the size of the test, they are specified such that the absence of Granger causality from $(Y,Z)\to X$ is established by construction. On the one hand, S1 represents the case of three independent stationary AR(1) processes with the same autoregressive coefficient and standard normal innovations. On the other hand, S2 has three independent GARCH(1,1) processes with the same parameter specification and independent innovations drawn from a standardized skewed Student's t distribution with 5 degrees of freedom and a skewness parameter of $-1.5$. Model S2 serves to emulate the behaviour of data from financial returns by inducing heavy tails and negative skewness in the innovations, which are well-known stylized facts in the literature related to financial returns \citep{cont2001}.

	\begin{table}[H]
		\centering
		\begin{tabular}{>{\centering\hspace{0pt}}m{0.179\linewidth}>{\centering\hspace{0pt}}m{0.225\linewidth}>{\centering\hspace{0pt}}m{0.225\linewidth}>{\centering\arraybackslash\hspace{0pt}}m{0.225\linewidth}}
			\toprule
			\textbf{DGP} & $T=100$ & $T=200$ & $T=500$ \\
			\midrule
			\multicolumn{4}{c}{$\boldsymbol{\tau=0.1}$} \\
			\midrule
			S1 & 0.056 & 0.062 & 0.056 \\
			S2 & 0.056 & 0.040 & 0.046 \\
			\midrule
			\multicolumn{4}{c}{$\boldsymbol{\tau=0.5}$} \\
			\midrule
			S1 & 0.042 & 0.040 & 0.046 \\
			S2 & 0.056 & 0.042 & 0.052 \\
			\midrule
			\multicolumn{4}{c}{$\boldsymbol{\tau=0.9}$} \\
			\midrule
			S1 & 0.064 & 0.050 & 0.042 \\
			S2 & 0.060 & 0.042 & 0.044 \\
			\bottomrule
		\end{tabular}
        \caption{Empirical size of the proposed M-vine Granger causality in expectiles test for $(Y,Z)\to X$ in each size assessment model.}
         \label{sizeres}
	\end{table}

Table~\ref{sizeres} presents the results of the empirical size assessment of the proposed test for both S1 and S2. Noticeably, for both DGP's the test manages to control the size around the predefined level of $\alpha=0.05$, meaning that the proposed M-vine test has a low probability of incorrectly rejecting the null hypothesis of no Granger causality in the $\tau$-expectile. Moreover, this conclusion holds for every considered value of $\tau$ and for every considered length $T$ of the series. For analyzing the power of the proposed test, we now introduce the power assessment models that we are going to employ.\\

\textbf{\underline{Power assessment models}}\\
\begin{flalign*}
& \begin{aligned}[t]
\makebox[2em][l]{\textbf{P1}}\; X_t &= 0.5 X_{t-1} + 0.2 Y_{t-1} + 0.2 Z_{t-1} + \eta_{x,t}, \\
Y_t &= 0.5 Y_{t-1} + \eta_{y,t}, \qquad Z_t = 0.5 Z_{t-1} + \eta_{z,t},
\end{aligned} &&
\end{flalign*}
where $(\eta_{x,t})_t, (\eta_{y,t})_t, (\eta_{z,t})_t$ are three independent white Gaussian noises. 

\medskip
\begin{flalign*}
& \begin{aligned}[t]
\makebox[2em][l]{\textbf{P2}}\; X_t &= 0.5 X_{t-1} + 5 Y_{t-1} Z_{t-1} + \eta_{x,t}, \\
Y_t &= 0.25 Y_{t-1} + \eta_{y,t}, \qquad Z_t = 0.25 Z_{t-1} + \eta_{z,t},
\end{aligned} &&
\end{flalign*}
where $(\eta_{x,t})_t, (\eta_{y,t})_t, (\eta_{z,t})_t$ are three independent white Gaussian noises. 

\medskip
\begin{flalign*}
& \begin{aligned}[t]
\makebox[2em][l]{\textbf{P3}}\; X_t &= 0.5 X_{t-1} + 5 Y_{t-1} Z_{t-1} + \eta_{x,t}, \\
Y_t &= \eta_{y,t}, \qquad Z_t = \eta_{z,t},
\end{aligned} &&
\end{flalign*}
where $(\eta_{x,t})_t, (\eta_{y,t})_t, (\eta_{z,t})_t$ are three independent white Gaussian noises. 

\medskip
\begin{flalign*}
& \begin{aligned}[t]
\makebox[2em][l]{\textbf{P4}}\; X_t &= 0.5 X_{t-1} + 2.5 Y_{t-1} Z_{t-1} + \varepsilon_{x,t}, \\
Y_t &= \varepsilon_{y,t}, \qquad Z_t = \varepsilon_{z,t}, \\
\varepsilon_{i,t} &= \sigma_{i,t}\,z_{i,t}, \\
\sigma_{i,t}^2 &= 0.01 + 0.08\,\varepsilon_{i,t-1}^2 + 0.87\,\sigma_{i,t-1}^2, \qquad i \in \{x,y,z\},
\end{aligned} &&
\end{flalign*}
where $(z_{x,t})_t, (z_{y,t})_t, (z_{z,t})_t$ are three independent sequences of i.i.d. random variables distributed as the standardized skewed Student's t distribution~$\text{sstd}(0,1,\nu=5,\xi=-1.5)$.\\

    These four DGP's are built so that there is joint Granger causality $(Y,Z)\to X$ and with or without single Granger causality $Y\to X$ and $Z\to X$. 
    Model P1 has a similar structure as S1, $Y_t$ and $Z_t$ are still independent AR(1) processes, whereas $X_t$ does not only depend linearly on its own lag, but also on the ones from $Y_t$ and $Z_t$. In P2, $Y_t$ and $Z_t$ are AR(1) processes with smaller autoregressive coefficient than in P1, 
    but it represents a more complex type of dependence driven by the interaction term in $X_t$ that involves the lags of $Y_t$ and $Z_t$. Model P3 is relatively similar to P2 with the only difference that $Y_t$ and $Z_t$ are both purely Gaussian white noises. This subtle difference makes P3 a compelling model, as one can analytically  show that there is no Granger causality in the mean from $Y\rightarrow X$ nor from $Z\rightarrow X$, however, there is joint Granger causality in the mean from $(Y,Z)\rightarrow X$ (for a formal proof of these facts, we refer to Appendix~\ref{app:gcprof}. See also Appendix~\ref{app:gcprof2}). Consequently, tests for pairwise Granger causality in the mean would correctly lead to the no-rejection of the null hypothesis of no Granger causality in the  mean (e.g. Table~\ref{p3sims} in Appendix~\ref{app:gcprof}), but, differently from multivariate tests, they could not be applied to detect the presence of the joint Granger causality. Lastly, P4 is built to analyze the power of the  test with simulated data that behaves as series of financial returns: specifically, it resembles S2 but in the case that $X$ depends on the lags of both $Y_t$ and $Z_t$ through a non-linear term.
\\

	\begin{table}[H]
		\centering
		\footnotesize
		\begin{tabular}{>{\centering\hspace{0pt}}m{0.179\linewidth}>{\centering\hspace{0pt}}m{0.225\linewidth}>{\centering\hspace{0pt}}m{0.225\linewidth}>{\centering\arraybackslash\hspace{0pt}}m{0.225\linewidth}}
			\toprule
			\textbf{DGP} & $T=100$ & $T=200$ & $T=500$ \\
			\midrule
			\multicolumn{4}{c}{$\boldsymbol{\tau=0.1}$} \\
			\midrule
			P1 & 0.420 & 0.586 & 0.882 \\
			P2 & 0.640 & 0.924 & 0.998 \\
			P3 & 0.676 & 0.934 & 1.000 \\
			P4 & 0.796 & 0.908 & 0.982 \\
			\midrule
			\multicolumn{4}{c}{$\boldsymbol{\tau=0.5}$} \\
			\midrule
			P1 & 0.516 & 0.724 & 0.978 \\
			P2 & 0.288 & 0.594 & 0.942 \\
			P3 & 0.296 & 0.610 & 0.972 \\
			P4 & 0.944 & 0.988 & 1.000 \\
			\midrule
			\multicolumn{4}{c}{$\boldsymbol{\tau=0.9}$} \\
			\midrule
			P1 & 0.412 & 0.624 & 0.886 \\
			P2 & 0.652 & 0.934 & 1.000 \\
			P3 & 0.672 & 0.958 & 1.000 \\
			P4 & 0.900 & 0.986 & 1.000 \\
			\bottomrule
		\end{tabular}
        \caption{Empirical power of the proposed M-vine Granger causality in expectiles test for $(Y,Z)\to X$ in each power assessment model.}
        \label{powerres}
	\end{table}

The results for the empirical power across the four DGP's, three expectile levels, and three sample sizes are exhibited in Table~\ref{powerres}. The first noticeable pattern is that for each DGP and expectile level, the power is monotonically increasing as $T$ grows, explicitly illustrating the consistency of the M-vine test. In fact, for $T=500$ the test reaches the ideal power in most dependence scenarios and remarkably, even for $T=200$ the test has a considerably high power in some DGP's. The process for which our test has the lowest power in $T=500$ is P1. This result has two implications. On the one hand, it shows the sensitivity of the proposed test to the size of the autoregressive coefficients on the lags of the causing series. This has been well-documented for other tests for Granger causality, for instance in the simulation study of \cite{bouezmarni2024}. On the other hand, this is an example of a scenario where the single autoregressive coefficients for each causing series might be small enough so that, in practice with a finite sample size, they lead to pairwise tests not rejecting the null hypothesis of no Granger causality. Therefore, the proposed test provides a tool that can address such settings where the pairwise dependencies are too weak to be detected by  pairwise tests, but the joint influence is strong enough that our test can catch it and distinguish these cases from those in which there is indeed no Granger causality. To illustrate this, under the P1 specification, we compare the power of the proposed M-vine test for $(Y,Z)\to X$ with the ones of our pairwise M-vine Granger causality in expectiles tests and with the ones of the Granger causality  in quantiles tests from \cite{balcilar2016}\footnote{\cite{balcilar2016} introduces a Kernel-based Non-Parametric (KNP) test for Granger causality in quantiles that is able to test for individual causality between series. Using the R package ‘\textit{nonParQuantileCausality}’, we implement this test for each direction: $Y\to X$ and $Z\to X$.}. 

\begin{table}[H]
		\centering
		\footnotesize
		\begin{tabular}{>{\centering\hspace{0pt}}m{0.25\linewidth}>{\centering\hspace{0pt}}m{0.225\linewidth}>{\centering\hspace{0pt}}m{0.225\linewidth}>{\centering\arraybackslash\hspace{0pt}}m{0.225\linewidth}}
			\toprule
			\textbf{Test} & $T=100$ & $T=200$ & $T=500$ \\
			\midrule
			\multicolumn{4}{c}{$\boldsymbol{\tau=0.1}$} \\
			\midrule
			M-vine $(Y,Z)\to X$ & 0.420 & 0.586 & 0.882 \\
            \midrule
            M-vine  $Y \to X$ & 0.196 & 0.356 & 0.650 \\
			M-vine  $Z \to X$ & 0.280 & 0.388 & 0.530 \\
			KNP  $Y \to X$ & 0.041 & 0.048 & 0.084 \\
			KNP  $Z \to X$ & 0.033 & 0.040 & 0.094 \\
			\midrule
			\multicolumn{4}{c}{$\boldsymbol{\tau=0.5}$} \\
			\midrule
			M-vine $(Y,Z)\to X$ & 0.516 & 0.724 & 0.978 \\
            \midrule
            M-vine  $Y \to X$ & 0.300 & 0.468 & 0.740 \\
			M-vine  $Z \to X$ & 0.344 & 0.464 & 0.720 \\
			KNP  $Y \to X$ & 0.154 & 0.340 & 0.709 \\
			KNP $Z \to X$ & 0.180 & 0.344 & 0.705 \\
			\midrule
			\multicolumn{4}{c}{$\boldsymbol{\tau=0.9}$} \\
			\midrule
			M-vine $(Y,Z)\to X$  & 0.412 & 0.624 & 0.886 \\
            \midrule 
            M-vine  $Y \to X$ & 0.264 & 0.388 & 0.620 \\
			M-vine  $Z \to X$ & 0.268 & 0.352 & 0.560 \\
			KNP $Y \to X$ & 0.000 & 0.000 & 0.082 \\
			KNP $Z \to X$ & 0.000 & 0.014 & 0.052 \\
			\bottomrule
		\end{tabular}
        \caption{Empirical power comparison for model P1: M-vine Granger causality in expectiles test for $(Y,Z)\to X$ versus the pairwise 
        M-vine Granger causality in expectiles test and versus the 
        Kernel-based Non-Parametric (KNP) tests for Granger causality in quantiles of \cite{balcilar2016}. Number of simulations: $S = 500$.}
        \label{powercomp}
	\end{table}

Table~\ref{powercomp} exhibits the results from the power comparison between the three tests. When $T=500$ and $\tau=0.5$ (which means Granger causality in the mean for the M-vine test and in the median for the quantile-based test), the proposed joint test has a notably higher power. In fact, our joint test reaches an almost ideal power, whereas the power of the pairwise tests are in the vicinity of $0.70$. When we move into the tails, the joint version of our test maintains power relatively closer to $0.90$ for $T=500$. The pairwise M-vine test has a power that increases with sample size also for extreme values of $\tau$, keeping the same pattern as in the case of the mean (although with lower values). In contrast, KNP test exhibits a power drastically close to or below $0.05$ (always below $0.10$) across all sample sizes, meaning that it is failing to detect a causal relationship that is present by construction.

\indent In summary, these results highlight the two aforementioned merits of the proposed methodology. Firstly, the joint nature of the proposed test captures the combined predictive contribution of $Y$ and $Z$, whereas the pairwise tests yield considerably weaker evidence as the single dependencies are too small to be detected individually. Secondly, as shown from the comparison between the pairwise tests, the copula-based estimation tends to be more efficient than the kernel-based counterpart when working with finite samples, especially when focusing on the tails where copulas excel at modeling and capturing dependencies.
\\

\subsection{Granger causality in the mean (i.e.~$\tau=1/2$) for $(Y,Z)\to X$: \\comparison with the classic Granger causality F-test} 

We repeat the simulation study of the previous section using the linear Granger causality in the mean F-test for $(Y,Z)\to X$, in order to compare its statistical properties with the ones of the M-vine test with $\tau=1/2$. Table~\ref{linear_granger} reports the empirical size and the power of the above mentioned linear test for $(Y,Z)\to X$. Regarding the size, the linear test controls the size around the nominal level when $T\geq 200$, similarly to the proposed M-vine test (compare the ``size part'' of Table~\ref{linear_granger} with Table~\ref{sizeres}). In terms of power, the results reveal evident differences between linear and non-linear settings. For the linear DGP P1, the F-test achieves an almost ideal power even with moderate sample sizes, showing a faster rate of convergence than the M-vine test with $\tau=1/2$. This is not a surprising result at all, given that the F-test is built on such linear models. However, for DGPs P2, P3, and P4, where the dependence structure is non-linear, the power of the F-test remains very low for every sample size, failing to increase with $T$. In contrast, the power of the M-vine test increases consistently with $T$ across all power assessment models, reaching near-ideal levels at $T=500$ (compare the ``power part'' of Table~\ref{linear_granger} with Table~\ref{powerres}).

\begin{table}[H]
	\centering
	\footnotesize
	\begin{tabular}{>{\centering\hspace{0pt}}m{0.179\linewidth}>{\centering\hspace{0pt}}m{0.225\linewidth}>{\centering\hspace{0pt}}m{0.225\linewidth}>{\centering\arraybackslash\hspace{0pt}}m{0.225\linewidth}}
		\toprule
		\textbf{DGP} & $T=100$ & $T=200$ & $T=500$ \\
		\midrule
		\multicolumn{4}{c}{\textbf{Size}} \\
		\midrule
		S1 & 0.028 & 0.052 & 0.060 \\
		S2 & 0.042 & 0.062 & 0.060 \\
		\midrule
		\multicolumn{4}{c}{\textbf{Power}} \\
		\midrule
		P1 & 0.810 & 0.988 & 1.000 \\
		P2 & 0.454 & 0.438 & 0.448 \\
		P3 & 0.350 & 0.338 & 0.340 \\
		P4 & 0.388 & 0.410 & 0.500 \\
		\bottomrule
	\end{tabular}
	\caption{Empirical size and power of the classical linear Granger causality in the mean F-test for $(Y,Z)\to X$.}
	\label{linear_granger}
\end{table}

These results confirm a very important aspect: a Granger causality test designed for linear dependence, such as the F-test, cannot detect non-linear causal relationships, whereas the proposed copula-based approach is able to capture them regardless of the underlying dependence structure (linear or non-linear).

\section{Empirical Applications}\label{empapp}

In this section, we present two applications of the proposed test in real data. We employ the M-vine Granger causality in expectiles test for studying the relationships among major stock market indices at both the global and the Asian regional level. Regarding the specification of the test, we work with the same setting as in Section~\ref{simstudy}, i.e., we set $N=200$, $T_0=\frac{T}{2}$, and $B=200$. We test for Granger causality at several positions of the distribution symmetrically around the mean by setting $\tau\in\{0.05,0.1,0.25,0.5,0.75,0.9,0.95\}$. For each application, we first perform the proposed test between all possible pairs of indices to assess individual Granger causality, and then we allow for pairs of indices to jointly cause an individual index, exploiting the multivariate feature of the tool. From the first application, we found that S\&P 500 has a pivotal role in terms of transmitting causality to both the FTSE 100 and the Nikkei 225 across the entire distribution, which comes as an expected result given the international relevance of the U.S. market. Moreover, our expectile-based test reveals that the FTSE causes the S\&P 500 in the sense of Granger only in the tails of the distribution, a pattern that would be invisible to a purely mean-based framework. Whereas in the second application, we find a more novel result: specifically, the test reveals that there is no Granger causality between pairs of the three Asian indices, but the Nikkei and the Shanghai Composite jointly Granger cause the Hang Seng Index in its left tail, empirically illustrating the advantage of the multivariate nature of the proposed test.

 \subsection{Causality among Global Stock Markets}

As first application, we study Granger causality among three globally relevant stock markets: the U.S., the U.K., and Japan. To this end, we work with data from their corresponding indices: the S\&P 500 (SP), the FTSE 100 (FTSE), and the Nikkei 225 (NK). The data consists of daily observations between September 2012 and October 2014 retrieved from Yahoo Finance. Notice that the length of the sample is $T=500$ in order to be consistent with the simulation study in Section~\ref{simstudy}. In order to make the data stationary, we transform stock prices into log-returns as $r_t=(\log(s_t)-\log(s_{t-1}))*100$, where $s_t$ is the stock price at time $t$.\\
	
\begin{table}[H]
	\centering
	\resizebox{\textwidth}{!}{%
		\footnotesize
		\begin{tabular}{lcccccccc}
			\toprule
			& \multicolumn{7}{c}{\textbf{$\tau$}} \\
			\cmidrule{2-8}
			\textbf{Direction} & 0.05 & 0.10 & 0.25 & 0.50 & 0.75 & 0.90 & 0.95 \\
			\midrule
			SP $\rightarrow$ FTSE   & 0.000$^{***}$ & 0.000$^{***}$ & 0.003$^{***}$ & 0.000$^{***}$ & 0.000$^{***}$ & 0.000$^{***}$ & 0.000$^{***}$ \\
			SP $\rightarrow$ NK     & 0.000$^{***}$ & 0.000$^{***}$ & 0.000$^{***}$ & 0.000$^{***}$ & 0.000$^{***}$ & 0.000$^{***}$ & 0.000$^{***}$ \\
			\midrule
			FTSE $\rightarrow$ SP   & 0.030$^{**}$\phantom{$^{*}$}  & 0.003$^{***}$ & 0.150\phantom{$^{***}$}  & 0.290\phantom{$^{***}$} & 0.500\phantom{$^{***}$} & 0.035$^{**}$\phantom{$^{*}$} & 0.025$^{**}$\phantom{$^{*}$}  \\
			FTSE $\rightarrow$ NK   & 0.000$^{***}$ & 0.000$^{***}$ & 0.000$^{***}$ & 0.000$^{***}$ & 0.000$^{***}$ & 0.000$^{***}$ & 0.000$^{***}$ \\
			\midrule
			NK $\rightarrow$ SP     & 0.340\phantom{$^{***}$} & 0.453\phantom{$^{***}$} & 0.470\phantom{$^{***}$}  & 0.525\phantom{$^{***}$}  & 0.248\phantom{$^{***}$} & 0.575\phantom{$^{***}$} & 0.233\phantom{$^{***}$} \\
			NK $\rightarrow$ FTSE   & 0.260\phantom{$^{***}$} & 0.472\phantom{$^{***}$} & 0.165\phantom{$^{***}$} & 0.145\phantom{$^{***}$} & 0.295\phantom{$^{***}$} & 0.470\phantom{$^{***}$} & 0.123\phantom{$^{***}$} \\
			\midrule
			(FTSE, NK) $\rightarrow$ SP  & 0.008$^{***}$ & 0.001$^{***}$& 0.135\phantom{$^{***}$} & 0.340\phantom{$^{***}$}& 0.128\phantom{$^{***}$} & 0.030$^{**}$\phantom{$^{*}$} & 0.022$^{**}$\phantom{$^{*}$}  \\
			(SP, NK) $\rightarrow$ FTSE   & 0.000$^{***}$ & 0.000$^{***}$ & 0.000$^{***}$ & 0.000$^{***}$ & 0.000$^{***}$ & 0.000$^{***}$ & 0.000$^{***}$ \\
			(SP, FTSE) $\rightarrow$ NK   & 0.000$^{***}$ & 0.000$^{***}$ & 0.000$^{***}$ & 0.000$^{***}$ & 0.000$^{***}$ & 0.000$^{***}$ & 0.000$^{***}$ \\
			\bottomrule
	\end{tabular}}
	\caption{$p$-values for the Granger causality in expectiles test between FTSE, NK and SP. $^{*}p<0.10$, $^{**}p<0.05$, $^{***}p<0.01$.}
    \label{appres}
\end{table}

	From the first panel of Table~\ref{appres}, the results show that the S\&P500 causes both the FTSE and NK individually across the entire distribution. This reflects the fact that among the three markets, the U.S. stock market is the one with the highest relevance worldwide, and therefore, acts as a transmitter of causality to the others. The second panel illustrates that the FTSE also causes the NK in the sense of Granger across most of the distribution, but to a lower extent than the S\&P500. Furthermore, it shows that the U.K. stock market has no predictive power in regular scenarios around the mean, however, extreme returns from the FTSE can help in predicting extreme movements in the S\&P500, as reflected by this symmetric pattern of Granger causality in both tails of the distribution. The third panel shows that the Japanese stock market acts purely as a receiver of causality, as it does not cause any of the other two indices at any part of the distribution. This illustrates how the Japanese market lacks predictive power for the other two indices that are arguably more relevant at a global level. Finally, the last panel of Table~\ref{appres} exhibits the results for joint Granger causality. One can notice that the pairs (SP,NK) and (SP,FTSE) cause the FTSE and NK, respectively, having a statistically significant improvement in their prediction across all values of $\tau$. This comes as an expected result as the S\&P500 causes each of these indices individually, hence, paired with another index the conclusion must hold. A similar argument applies to the pair (FTSE,NK) causing the S\&P500. The U.K. stock market already has enough predictive power in the tails of the distribution of the U.S. market by itself, then the addition of information stemming from the Japanese market, whose prediction power was already shown to be negligible across most of the distribution, does not change the results.\\

    Overall, these results present a coherent picture of global stock market interdependence in which the S\&P500 occupies a pivotal transmitting role across the entire return distribution, while NK acts mostly as a receiver. The tail-specific nature of the causality flow from FTSE to S\&P500, and the joint tail causality from (FTSE, NK) to S\&P500, highlight the value of our expectile-based tool over traditional mean-based approaches in capturing the full distributional structure of international return spillovers.\\

    \subsection{Causality among Asian Stock Markets}

For this second application, we shift the focus to three of the main Asian stock markets: mainland China, Hong Kong, and Japan. We examine causality between their respective indices: the Shanghai Composite (SSE), the Hang Seng Index (HSI), and the Nikkei 225 (NK). As in the previous application, the dataset consists of daily observations between September 2012 and October 2014 retrieved from Yahoo Finance, providing the same sample size of $T=500$. We also work with log-returns in order to achieve the stationarity of the three series.\\

    	\begin{table}[H]
		\centering
		\resizebox{\textwidth}{!}{%
			\footnotesize
			\begin{tabular}{lcccccccc}
				\toprule
				& \multicolumn{7}{c}{\textbf{$\tau$}} \\
				\cmidrule{2-8}
				\textbf{Direction} & 0.05 & 0.10 & 0.25 & 0.50 & 0.75 & 0.90 & 0.95 \\
				\midrule
				NK $\rightarrow$ HSI   & 0.080$^{*}$\phantom{$^{**}$}  & 0.080$^{*}$\phantom{$^{**}$}  & 0.145\phantom{$^{***}$} & 0.470\phantom{$^{***}$} & 0.365\phantom{$^{***}$} & 0.305\phantom{$^{***}$} & 0.370\phantom{$^{***}$} \\
				NK $\rightarrow$ SSE     & 0.410\phantom{$^{***}$} &0.520\phantom{$^{***}$}  &0.325\phantom{$^{***}$}  &0.345\phantom{$^{***}$}  &0.285\phantom{$^{***}$}  &0.340\phantom{$^{***}$}  &0.425\phantom{$^{***}$}  \\
				\midrule
				HSI $\rightarrow$ NK   & 0.440\phantom{$^{***}$}  &0.275\phantom{$^{***}$}  &0.330\phantom{$^{***}$}   &0.365\phantom{$^{***}$} &0.310\phantom{$^{***}$} &0.305\phantom{$^{***}$}  &0.360\phantom{$^{***}$}  \\
				HSI $\rightarrow$ SSE   & 0.150\phantom{$^{***}$} &0.370\phantom{$^{***}$} &0.255\phantom{$^{***}$}  &0.240\phantom{$^{***}$}  &0.235\phantom{$^{***}$}  &0.515\phantom{$^{***}$}  &0.195\phantom{$^{***}$}  \\
				\midrule
				SSE $\rightarrow$ NK     & 0.445\phantom{$^{***}$} & 0.370\phantom{$^{***}$} & 0.325\phantom{$^{***}$}  & 0.410\phantom{$^{***}$}  & 0.365\phantom{$^{***}$} & 0.305\phantom{$^{***}$} & 0.360\phantom{$^{***}$} \\
				SSE $\rightarrow$ HSI   & 0.340\phantom{$^{***}$} & 0.155\phantom{$^{***}$} & 0.260\phantom{$^{***}$} & 0.535\phantom{$^{***}$} & 0.450\phantom{$^{***}$} & 0.325\phantom{$^{***}$} & 0.275\phantom{$^{***}$} \\
				\midrule
				(NK, SSE) $\rightarrow$ HSI  & 0.005$^{***}$ & 0.025$^{**}$\phantom{$^{*}$} & 0.080$^{*}$\phantom{$^{**}$} & 0.545\phantom{$^{***}$}& 0.340\phantom{$^{***}$} & 0.235\phantom{$^{***}$} & 0.195\phantom{$^{***}$}  \\
				(HSI, NK) $\rightarrow$ SSE  & 0.225\phantom{$^{***}$} & 0.395\phantom{$^{***}$} & 0.305\phantom{$^{***}$} & 0.220\phantom{$^{***}$}& 0.305\phantom{$^{***}$} & 0.385\phantom{$^{***}$} & 0.575\phantom{$^{***}$} \\
				(HSI, SSE) $\rightarrow$ NK  & 0.325\phantom{$^{***}$} & 0.200\phantom{$^{***}$} & 0.210\phantom{$^{***}$} & 0.230\phantom{$^{***}$} & 0.335\phantom{$^{***}$} & 0.270\phantom{$^{***}$} & 0.185\phantom{$^{***}$} \\
				\bottomrule
		\end{tabular}}
		\caption{$p$-values for the Granger causality in expectiles test between HSI, NK and SSE. $^{*}p<0.10$, $^{**}p<0.05$, $^{***}p<0.01$.}
		\label{appres2}
	\end{table}

	Table \ref{appres2} reports the $p$-values for Granger causality in expectiles across all pairwise and joint directions among the Nikkei 225, Hang Seng Index, and Shanghai Composite. From the first three panels, there is evident absence of Granger causality at the pairwise level. None of the six pairwise combinations exhibit significant Granger causality across the entire distribution. Consequently, a standard pairwise analysis would therefore conclude that there are no spillover effects between these three Asian stock markets. However, the joint test provides a different conclusion. While (HSI, NK) $\rightarrow$ SSE and (HSI, SSE) $\rightarrow$ NK are still statistically insignificant for all $\tau$, the direction (NK, SSE) $\rightarrow$ HSI displays a strong level of joint causality in the left tail of the distribution. In fact, the null hypothesis of no Granger causality in this direction is rejected at the 1\% significance level for $\tau=0.05$ and at the 5\% level for $\tau=0.10$, whereas no evidence of causality is detected from the mean of the distribution until its right tail. This increase in predictability from (NK, SSE) $\rightarrow$ HSI as $\tau$ decreases, is suggestive of an increase in dependence under market stress. This pattern of no individual causality but significant joint causality constitutes a clear empirical instance where standard pairwise testing misses a genuine spillover mechanism at the tails. Furthermore, this provides evidence of a real situation in which the scenario from model P1 from Section~\ref{simstudy} occurs, highlighting the importance of being able to test for joint Granger causality. From an economic standpoint, this implies that neither the NK nor the SSE provide enough information to anticipate downturns of the HSI, despite the HSI being one of the most prominent indices of this region. Nevertheless, they do so jointly, suggesting that when both mainland China and Japan are simultaneously in distress they constitute a comprehensive Asian risk signal that carries predictive power for Hong Kong's downside risk.\\

    Lastly, it is worth noting that NK $\rightarrow$ HSI exhibits a marginal significance at the 10\% level for $\tau=0.05$ and $\tau=0.10$, which despite our interpretation of being not significant, it could be construed as evidence of a weak individual Granger causality from Japan to Hong Kong at these corresponding expectile levels. Even under this interpretation, the joint test together with SSE shifts the significance level with which the null is rejected from a borderline 10\% to a clear 1\%, indicating that the Shanghai Composite index contributes complementary information that can make this predictability become fully evident.

\section{Conclusion}\label{conc}

Traditional Granger causality analysis focuses on testing whether past values of one time series contain predictive information about another, typically within a linear conditional mean framework. However, many real-world systems — particularly in finance, economics, and environmental science — exhibit complex, asymmetric, and non-linear dependence structures that cannot be fully captured by standard linear models. To address these limitations, recent research has extended Granger causality concepts to the expectile domain, providing a flexible framework for investigating causality across different parts of a conditional distribution rather than merely the mean. In this paper, by generalizing the mean measure of \cite{song2018} we introduced a model-free measure of Granger causality in expectiles, and using this novel measure we proposed a testing procedure based on the M-vine copula models from \cite{beare2015} that can account for multivariate Granger causality under non-linear and non-Gaussian settings. Under some (standard) regularity conditions, we established the strong consistency of the proposed test statistic. Furthermore, by means of a simulation study, we showed that the proposed test controls the size around the nominal significance level while having a considerably high power for a wide array of data generating processes, even for medium sample sizes. Lastly, we employed the M-vine Granger causality in expectiles test in two empirical applications testing for Granger causality between global major stock market indices and Asian stock market indices.\\

\bibliography{ms}

\pagebreak
\appendix
\section{Appendix}
\subsection{Theoretical results}
By standard arguments (e.g. \cite[Sec. 5.7]{vanderVaart1998}), we can prove the following theorem.
\begin{theorem}[Strong consistency of the empirical $\tau$-expectile]\label{th-strong-consistency}
	Let $X$ be a real random variable with ${\mathbb E}[X^2]<\infty$ and let 
    $X_1,X_2,\dots$ be i.i.d.\ real random variables with the same distribution of $X$ and fix 
    $\tau\in (0,1)$. Denote by $\mu_{\tau}$ the $\tau$-expectiles of $X$, i.e. the unique minimizer (over $m$) of 
	\[
	Q_\tau(m)={\mathbb E}\Big[{\mathcal R}_\tau(X-m)\Big]\,.
	\]
	Denote by  $\widehat{\mu}_{\tau,N}$ the empirical  
    $\tau$-expectile, i.e. the random variable $\widehat{\mu}_{\tau,N}(X_1,\dots,X_N)$, where 
    $\widehat{\mu}_{\tau,N}(x_1,\dots,x_N)$ is the unique minimizer (over $m$)  of 
	\[
	Q_{\tau,N}(x_1,\dots,x_N, m)=\frac{1}{N}\sum_{i=1}^N {\mathcal R}_\tau(x_i-m).
	\]
	Then we have $\widehat{\mu}_{\tau,N}\xrightarrow[N\to\infty]{a.s.} \mu_\tau$.
	
\end{theorem}
\begin{proof} We split the proof in some steps. The first two steps are devoted to verify that $\mu_\tau$ and $\hat\mu_{\tau,N}$ are well-defined; while the other steps prove the strong consistency of $\hat\mu_{\tau,N}$. We begin by showing the differentiability of $m\mapsto Q_{\tau,N}(x_1,\dots,x_N,m)$ and $m\mapsto Q_\tau(m)$.
\\

	For each $x\in{\mathbb R}$, the function 
    \begin{equation*}
    \begin{split}
    m\mapsto {\mathcal R}_\tau(x-m)&=|\,\tau - \mathbb{I}_{\{(x-m)<0\}} \,|\, (x-m)^2 \\
&=\begin{cases}
\tau\,(x-m)^2\quad&\mbox{on } m\leq x\\
(1-\tau)\,(x-m)^2\quad&\mbox{on } m>x
\end{cases}
\end{split}
\end{equation*}
    is differentiable everywhere in $\mathbb R$ 	and, for each $m\in \mathbb R$, we have 
	\begin{equation}\label{eq:psi_def}
    \begin{split}
\frac{\partial}{\partial m}{\mathcal R}_\tau(x-m)	=	\psi_\tau(x,m)
		&=
		2\Big(\tau\mathbf 1_{\{(x-m)\geq 0\}}+(1-\tau)\mathbf 1_{\{(x-m)< 0 \}}\Big)(m-x)\\
        &=\begin{cases}
        2\tau\, (m-x)\quad&\mbox{if } m\leq x\\
        2(1-\tau)\,(m-x)\quad&\mbox{if } m>x\,.
        \end{cases}
        \end{split} 
	\end{equation}
    (Note that piecewise defined functions are often not differentiable in the points where the definition changes, but in the case of $m\mapsto {\mathcal R}_\tau(x-m)$, we have the term $(x-m)^2$ that makes the function differentiable also in $m=x$.) 
		It immediately follows that, for all $x_1,\dots, x_n$, the real function $m\mapsto Q_{\tau,N}$ is differentiable with derivative 
        \begin{equation}\label{def:zeta}
        Q'_{\tau,N}(m)	= z_{\tau,N}(x_1,\dots,x_N,m)=
\frac{1}{N}\sum_{i=1}^N \psi_\tau(x_i,m)\,.
        \end{equation}
        Moreover, setting 
	\[
	D_h(x,m)=\frac{{\mathcal R}_\tau(x-(m+h))-{\mathcal R}_\tau(x-m)}{h}\,,
	\]
	by the mean value theorem, we have  $D_h(x,m)=
	\psi_\tau\big(x,m+y \,h\big)$ for some $y=y(x,h) \in (0,1)$ 
	and so, for  $|h|\le 1$, we have 
	\[
	|D_h(x,m)|
	=
	\big|\psi_\tau\big(x,m + y\, h\big)\big|
	\le 2\big(|m +y\, h|+|x|\big)
	\le 2\big(|m|+1+|x|\big).
	\]
	Therefore, since $|D_h(X,m)|\leq 2(|m|+1+|X|)$ for $|h|\leq 1$ and ${\mathbb E}[\,|X|\,]<+\infty$, we can apply the dominated convergence theorem and obtain, for each $m$, 
	\[Q'_\tau(m)=
	\lim_{h\to 0} \frac{Q_\tau(m+h)-Q_\tau(m)}{h}=
	\lim_{h\to 0}{\mathbb E}[D_h(X,m)]
	=
	{\mathbb E}\!\left[\lim_{h\to 0}D_h(X,m)\right]
	=
	{\mathbb E}[\psi_\tau(X,m)]\,.
	\]
    
	We now turn to the existence, uniqueness and characterization of the minimizers. The real functions 
	 $m\mapsto Q_{\tau,n}(x_1,\dots,x_N, m)$ and $m\mapsto Q_{\tau}(m)$ are differentiable (as shown above) and strictly convex (since empirical average/expectation of strictly convex functions). Therefore, they have unique minimizers $\hat{\mu}_{\tau,N}(x_1,\dots,x_N)\in{\mathbb R}$ and $\mu_\tau\in {\mathbb R}$, which are characterized by 
     $$
     z_{\tau,N}(x_1,\dots,x_N,\hat{\mu}_{\tau,N}(x_1,\dots,x_N))=0
     \qquad\mbox{and}\qquad
     z_\tau(\mu_\tau)=0
     $$
    where $z_{\tau,N}(x_1,\dots,x_N,m)$ is defined in \eqref{def:zeta} and 
	$z_\tau(m)=Q'_\tau(m)={\mathbb E}[\psi_\tau(X,m)]$ and both of them 
    are continuous and strictly increasing functions of $m$.	
    We then set 
	\[Z_{\tau,N}(m)=
z_{\tau,N}(X_1,\dots,X_N,m)=\frac{1}{N}\sum_{i=1}^N \psi_\tau(X_i,m)
\qquad\mbox{and}\qquad 
\hat\mu_{\tau,N}=\hat{\mu}_{\tau,N}(X_1,\dots,X_N)\,. 
	\]
(Note that, since $m\mapsto z_{\tau,N}(x_1,\dots,x_N)$ is continuous and strictly increasing and, for each fixed $m$,  the function $x\mapsto \psi_\tau(x,m)$ is continuous (and strictly decreasing), we have that $\hat\mu_{\tau,N}(x_1,\dots,x_N)$ is Borel-measurable and so $\hat\mu_{\tau,N}$ is a well-defined real random variable.)\\

A key aspect for the upcoming convergence argument is the following Lipschitz property of $m\mapsto z_{\tau,N}(x_1,\dots,x_N,m)$ and $m\mapsto z_\tau(m)$. 
		For all $m,\tilde m$ and $x$, we have 
	\[
	|\psi_\tau(x,m)-\psi_\tau(x,\tilde m)|
	\leq 2|m-\tilde m|
	\]
	and so, for all $m,\tilde m$ and $x_1,\dots,x_N$, we have 
	\[
	|z_{\tau,N}(x_1,\dots,x_N,m) - z_{\tau,N}(x_1,\dots,x_N,\tilde m)|
	\leq 2\,|m-\tilde m|,
	\qquad
	|z_\tau(m) - z_\tau(\tilde m)|
	\leq 2\,|m-\tilde m|.
	\]

We can now derive the uniform strong law of large numbers on compacts for $Z_{\tau, N}$.
		 Fix any $M>0$. For all $m\in[-M,M]$ and all $x\in{\mathbb R}$, we have 
	\[
	|\psi_\tau(x,m)|\le 2(|m|+|x|)\le 2(M+|x|)\,.
	\]
Hence, since  $\mathbb{E}[\,|X|\,]<+\infty$, 
	the strong law of large numbers applies point-wise in $[-M,M]$, that is, for each fixed $m\in [-M,M]$, we have 
	\begin{equation}\label{eq:pslln}
	{Z}_{\tau,N}(m)\stackrel{a.s.}
	\longrightarrow
	z_\tau(m)
	\qquad\text{as } N\to +\infty\,.
	\end{equation}
	In addition, the above limit holds true also uniformly on $[-M,M]$. Indeed, take an arbitrary  $\delta>0$. Since $[-M,M]$ is compact,
    there exist finitely many grid points
	$m_1,\dots,m_K$ such that, for every
	$m \in [-M,M]$, there exists $k$ that satisfies $|m-m_k|\leq \delta$. Since the grid is finite and the strong law of large numbers holds true at each grid point, we get 
	\begin{equation}\label{eq:uslln-grid}
	\max_{1\leq k\leq K}
	|Z_{\tau,N}(m_k)-z_\tau(m_k)|
	\;\xrightarrow[N\to\infty]{a.s.}\; 0.
	\end{equation}
Moreover, we can write 
	\[
		|Z_{\tau,N}(m)-z_\tau(m)|
		\leq
		|Z_{\tau,N}(m)-{z}_{\tau,N}(m_k)|
		+
		|Z_{\tau,N}(m_k)-z_\tau(m_k)|
		+
		|z_\tau(m_k)-z_\tau(m)|.
	\]
Hence, using the Lipschitz property established above (with $\tilde m=m_k$ such that $|m-m_k|\leq \delta$), we find 
	\[
	|Z_{\tau,N}(m)-z_\tau(m)|
	\leq
	|Z_{\tau,N}(m_k)-z_\tau(m_k)|
	+ 4\delta\,,
	\]
	which implies 
	$$
	\sup_{m\in[-M,M]}
	|Z_{\tau,N}(m)-z_\tau(m)|
	\leq
	\max_{1\le k\le K}
	|Z_{\tau,N}(m_k)-z_\tau(m_k)|
	+ 4\delta.
	$$
	Taking $\limsup_{N\to\infty}$ and using \eqref{eq:uslln-grid}, we get 
	\[
	\limsup_{N\to\infty}
	\sup_{m\in[-M,M]}
	|Z_{\tau,N}(m)-z_\tau(m)|
	\le 4\delta
	\qquad\text{a.s.}
	\]
	Since $\delta>0$ is arbitrary, we conclude that
	\begin{equation}\label{eq:uslln}
	\sup_{m\in[-M,M]}
	|Z_{\tau,N}(m)-z_\tau(m)|
	\;\xrightarrow[N\to\infty]{a.s.}\; 0.
	\end{equation}
	
	Lastly, we combine the uniform convergence and the characterization of the minimizers to obtain the desired result. Choose $M>0$ such that $|\mu_\tau|<M$. Taking an arbitrary $\epsilon>0$
	such that $\mu_\tau\pm\epsilon\in (-M,M)$, we have $z_\tau(\mu_\tau-\epsilon)<0$ and $z_\tau(\mu_\tau+\epsilon)>0$ (since $z_\tau$ is strictly increasing with $z_\tau(\mu_\tau)=0$).
    Moreover, since
\begin{equation*}
\begin{split}
Z_{\tau,N}(\mu_\tau-\epsilon)&\leq 
|Z_{\tau,N}(\mu_\tau-\epsilon)-z_\tau(\mu_\tau-\epsilon)|+z_\tau(\mu_\tau-\epsilon)
\leq
\sup_{m\in[-M,M]}
	|Z_{\tau,N}(m)-z_\tau(m)|+z_\tau(\mu_\tau-\epsilon)\\
Z_{\tau,N}(\mu_\tau+\epsilon)&\geq 
z_\tau(\mu_\tau+\epsilon)-
	|Z_{\tau,N}(\mu_\tau+\epsilon)-z_\tau(\mu_\tau+\epsilon)|\geq
z_\tau(\mu_\tau+\epsilon)-
\sup_{m\in[-M,M]}
	|Z_{\tau,N}(m)-z_\tau(m)|\,,
\end{split}
\end{equation*}
	by means of \eqref{eq:uslln}, we get that almost surely
	\[
	Z_{\tau,N}(\mu_\tau-\epsilon)<0
	\quad\text{and}\quad
	Z_{\tau,N}(\mu_\tau+\epsilon)>0 \qquad \mbox{for sufficiently large } N.
	\]
    (Note that it would be enough to use the point-wise strong law of large numbers \eqref{eq:pslln} at the two points $\mu_\tau\pm \epsilon\in [-M,M]$, but we decided to state and apply the more elegant version \eqref{eq:uslln}.) 
	Since $\hat\mu_{\tau,N}$ is the unique zero of $m\mapsto Z_{\tau,N}(m)$ which is continuous and strictly increasing, it follows that, almost surely, the (unique) zero-point $\hat\mu_{\tau,N}$ belongs to $(\mu_\tau-\epsilon,\mu_\tau+\epsilon)$ for sufficiently large $N$, that is almost surely 
	\[
	|\widehat\mu_{\tau,N}-\mu_\tau|<\epsilon\qquad \mbox{for sufficiently large } N.
	\]
	Since $\epsilon>0$ is arbitrarily sufficiently small, we can conclude that $\widehat\mu_{\tau,N}\stackrel{a.s.}\longrightarrow \mu_\tau$.
\end{proof}

\begin{corollary}\label{cor-1} If $(Y_t)_t$ is a Markov square-integrable stationary stochastic process, then, for each fixed $T$, we have 
$$
\frac{1}{T} \sum_{t=0}^T {\mathcal R}_\tau(Y_t-\mu_{\tau,N}(Y_t|Y_{t-1}))
\stackrel{a.s.}\longrightarrow 
\frac{1}{T} \sum_{t=0}^T {\mathcal R}_\tau(Y_t-\mu_{\tau}(Y_t|Y_{t-1}))
\quad\mbox{for } N\to +\infty\,.
$$
where $\mu_{\tau}(Y_t|Y_{t-1})$ and $\mu_{\tau,N}(Y_t|Y_{t-1})$ are defined as in the previous theorem taking as the distribution of $X$ the conditional distribution of $Y_t$ given $Y_{t-1}$.  
\end{corollary}
\begin{proof} It is an immediate consequence of the above theorem and the continuity of the function $m\mapsto {\mathcal R}_\tau(x-m)$.  
\end{proof}

\begin{corollary}\label{cor-2} If $(Y_t)_t$ is a Markov square-integrable ergodic stationary process, then, taking first $T\to +\infty$ and then $N\to +\infty$, we have 
$$
\frac{1}{T} \sum_{t=0}^T {\mathcal R}_\tau(Y_t-\mu_{\tau,N}(Y_t|Y_{t-1}))
\stackrel{a.s.}\longrightarrow 
E[\, {\mathcal R}_\tau(Y_t-\mu_{\tau}(Y_t|Y_{t-1}))\,]\,.
$$
where $\mu_{\tau}(Y_t|Y_{t-1})$ and $\mu_{\tau,N}(Y_t|Y_{t-1})$ are defined as in the previous theorem taking as the distribution of $X$ the conditional distribution of $Y_t$ given $Y_{t-1}$.  
\end{corollary}
\begin{proof} By the strong law of large numbers for ergodic stationary processes, we have
$$
\frac{1}{T} \sum_{t=0}^T {\mathcal R}_\tau(Y_t-\mu_{\tau}(Y_t|Y_{t-1}))
\stackrel{a.s.}\longrightarrow 
E[\, {\mathcal R}_\tau(Y_t-\mu_{\tau}(Y_t|Y_{t-1}))\,]\quad\mbox{ for } T\to +\infty\,.
$$
Hence, with probability one, for each $\epsilon>0$, there exists $T^*$ (depending on $\omega$ and $\epsilon$) large enough such that 
\begin{equation*}
\begin{split}
&\left|\frac{1}{T} \sum_{t=0}^T {\mathcal R}_\tau(Y_t-\mu_{\tau,N}(Y_t|Y_{t-1})) -
E[\, {\mathcal R}_\tau(Y_t-\mu_{\tau}(Y_t|Y_{t-1}))\,]\right|\leq \\
&\left|\frac{1}{T} \sum_{t=0}^T {\mathcal R}_\tau(Y_t-\mu_{\tau}(Y_t|Y_{t-1})) -
E[\, {\mathcal R}_\tau(Y_t-\mu_{\tau}(Y_t|Y_{t-1}))\,]\right|+\\
&\left|\frac{1}{T} \sum_{t=0}^T {\mathcal R}_\tau(Y_t-\mu_{\tau,N}(Y_t|Y_{t-1})) -
\frac{1}{T} \sum_{t=0}^T {\mathcal R}_\tau(Y_t-\mu_{\tau}(Y_t|Y_{t-1})) 
\right|\leq\\
&\frac{\epsilon}{2} + \left|\frac{1}{T} \sum_{t=0}^T {\mathcal R}_\tau(Y_t-\mu_{\tau,N}(Y_t|Y_{t-1})) -
\frac{1}{T} \sum_{t=0}^T {\mathcal R}_\tau(Y_t-\mu_{\tau}(Y_t|Y_{t-1})) 
\right|\quad\forall T\geq T^*\,.
\end{split}
\end{equation*}
Then, by the previous corollary, for each fixed $T$, there exists $N^*$ sufficiently large (depending on $T$) 
such that 
$$
\left|\frac{1}{T} \sum_{t=0}^T {\mathcal R}_\tau(Y_t-\mu_{\tau,N}(Y_t|Y_{t-1})) -
\frac{1}{T} \sum_{t=0}^T {\mathcal R}_\tau(Y_t-\mu_{\tau}(Y_t|Y_{t-1})) 
\right|\leq \frac{\epsilon}{2}\quad\forall N\geq N^*\,.
$$
Summing up, with probability one, for each $\epsilon>0$, there exists $T^*$ (depending on $\omega$ and $\epsilon$) such that, for each $T\geq T^*$, there exists $N^*$ (depending on $T$) such that 
$$
\left|\frac{1}{T} \sum_{t=0}^T {\mathcal R}_\tau(Y_t-\mu_{\tau,N}(Y_t|Y_{t-1})) -
E[\, {\mathcal R}_\tau(Y_t-\mu_{\tau}(Y_t|Y_{t-1}))\,]\right|\leq \epsilon\,.
$$
\end{proof}

\subsection{Power-assessment model P3: non-pairwise but joint Granger causality in the mean}\label{app:gcprof}

Recall that model P3 is
\begin{flalign*}
& \begin{aligned}[t]
\makebox[2em][l]{\textbf{P3}}\; X_t &= 0.5 X_{t-1} + 5 Y_{t-1} Z_{t-1} + \eta_{x,t}, \\
Y_t &= \eta_{y,t}, \qquad Z_t = \eta_{z,t},
\end{aligned} &&
\end{flalign*}
where $(\eta_{x,t})_t, (\eta_{y,t})_t, (\eta_{z,t})_t$ are three independent white Gaussian noises.\\

We will show analytically that there is no individual Granger causality in the mean from $Y\rightarrow X$ nor from $Z\rightarrow X$, but there is joint Granger causality in the mean from $(Y,Z)\rightarrow X$. To this end, we start with $Y\rightarrow X$. Taking conditional expectations of $X_t$ given $X_{t-1}$ yields
$$
\mathbb{E}[X_t|X_{t-1}]=0.5 X_{t-1}+5\mathbb{E}[Y_{t-1} Z_{t-1}|X_{t-1}]+\mathbb{E}[\eta_{x,t}|X_{t-1}],
$$

\noindent by the independence of $[Y_{t-1},Z_{t-1}]$ from $X_{t-1}$, and the fact that both have zero mean as well as $\eta_{x,t}$, we get 
$$
\mathbb{E}[X_t|X_{t-1}]=0.5 X_{t-1}.
$$

\noindent Similarly, we now take conditional expectations of $X_t$ given $X_{t-1}$ and $Y_{t-1}$ yielding
$$
\mathbb{E}[X_t|X_{t-1},Y_{t-1}]=0.5 X_{t-1}+5Y_{t-1} \mathbb{E}[Z_{t-1}|X_{t-1},Y_{t-1}]+\mathbb{E}[\eta_{x,t}|X_{t-1},Y_{t-1}].
$$
Given that $Z_{t-1}$ and $\eta_{x,t}$ are independent of $[X_{t-1},Y_{t-1}]$, and both have mean zero, we have
$$
\mathbb{E}[X_t|X_{t-1},Y_{t-1}]=0.5 X_{t-1}=\mathbb{E}[X_t|X_{t-1}].
$$
Since both conditional expectations are identical, the mean squared prediction error using the lags of both $X_{t-1}$ and $Y_{t-1}$ is the same as the one using only $X_{t-1}$, implying that there is no Granger causality in the mean from $Y$ to $X$. The exact same argument holds for the case of $Z\rightarrow X$. \\

Now, we will show that there is Granger causality in the mean from $(Y,Z)\rightarrow X$. We take conditional expectation of $X_t$ given $[X_{t-1}$, $Y_{t-1},Z_{t-1}]$, considering that $\eta_{x,t}$ is a Gaussian white noise, independent of $(\eta_{y,t-1},\eta_{z,t-1})$, we get
$$
\mathbb{E}[X_t|X_{t-1},Y_{t-1},Z_{t-1}]=0.5 X_{t-1}+5Y_{t-1} Z_{t-1}
$$
Since this conditional expectation differs from $\mathbb{E}[X_t|X_{t-1}] = 0.5X_{t-1}$, we now compute the mean squared prediction errors explicitly to confirm the improvement. The mean squared prediction error using only $X_{t-1}$ is
$$
\mathbb{E}\left[(X_t - \mathbb{E}[X_t|X_{t-1}])^2\right] = \mathbb{E}\left[(5Y_{t-1}Z_{t-1} + \eta_{x,t})^2\right].
$$
Expanding and using the independence of $Y_{t-1}$, $Z_{t-1}$, and $\eta_{x,t}$, as well as $\mathbb{E}[Y_{t-1}^2]=\mathbb{E}[Z_{t-1}^2]=\mathbb{E}[\eta_{x,t}^2]=1$, we obtain
$$
\mathbb{E}\left[(X_t - \mathbb{E}[X_t|X_{t-1}])^2\right] = 25\,\mathbb{E}[Y_{t-1}^2]\,\mathbb{E}[Z_{t-1}^2] + \mathbb{E}[\eta_{x,t}^2] = 25 + 1 = 26.
$$
The mean squared prediction error using $X_{t-1}$, $Y_{t-1}$, and $Z_{t-1}$ is
$$
\mathbb{E}\left[(X_t - \mathbb{E}[X_t|X_{t-1},Y_{t-1},Z_{t-1}])^2\right] = \mathbb{E}\left[\eta_{x,t}^2\right] = 1.
$$
Since $1 < 26$, the mean squared prediction error strictly decreases when both $Y_{t-1}$ and $Z_{t-1}$ are included alongside the past of $X$. Therefore, there is Granger causality in the mean from $(Y,Z)$ to $X$. 
\\

We also provide in Table~\ref{p3sims} the results of a simulation study, where we have applied 
the pairwise classic linear Granger causality F-test based on restricted and unrestricted autoregressive models as well as the pairwise version of our M-vine test with $\tau=1/2$.\\

	\begin{table}[H]
		\centering
		\footnotesize
		\begin{tabular}{>{\centering\hspace{0pt}}m{0.35\linewidth}>{\centering\arraybackslash\hspace{0pt}}m{0.2\linewidth}}
			\toprule
			\textbf{Test} & $T=500$ \\
			\midrule M-vine test $Y \to X$ & 0.072 \\
			M-vine test $Z \to X$ & 0.048 \\	
            Linear F-test $Y \to X$ &  0.244\\
			Linear F-test $Z \to X$ & 0.246 \\
			\bottomrule
		\end{tabular}
		\caption{Model P3: empirical rejection rates for the pairwise classic linear F-test for Granger causality in the conditional mean and 
        for the pairwise  M-vine test in the mean ($\tau=1/2$). We have taken $\alpha = 0.05$ and $S = 500$ simulations.}
		\label{p3sims}
	\end{table}

The results in Table~\ref{p3sims} show that the M-vine test has a rejection rate close to the nominal significance level of $\alpha = 0.05$ for both directions, demonstrating the absence of individual Granger causality from both $Y \to X$ and $Z \to X$. In contrast, the linear F-test exhibits clear size distortions stemming from the fact that this test is built under the assumption of a linear model, which is evidently misspecified given that P3 features an interaction term. These results not only confirm the analytical findings from this section, but also further illustrate the advantage of our copula-based approach in settings with non-linear dependence.

\subsection{Power-assessment model P3: pairwise and joint Granger causality in $\tau$-expectiles, with $\tau\neq 1/2$} 
\label{app:gcprof2}

We can deepen the analysis of model P3, showing  that, for every $\tau\in(0,1)$, there is joint Granger causality in the $\tau$-expectile from $(Y,Z)\to X$; while there is only pairwise Granger causality in the $\tau$-expectile from $Y\to X$ (and from $Z\to X$) if and only if $\tau\neq 1/2$.\\

\noindent 
By the model assumptions and the properties of the expectiles \citep{bellini2014}, we can write  
%
\[
\mu_\tau(X_t\mid X_{t-1})=0.5X_{t-1}+c_\tau,
\]
where $c_\tau=\mu_\tau(W_t)$ is a real constant depending on the distribution of $W_t=5Y_{t-1}Z_{t-1}+\eta_{x,t}$ and such that $c_{1/2}=0$.\\
\indent We now compute the conditional $\tau$-expectile of $X_t$ given $X_{t-1}$ and $Y_{t-1}$. We have 
$$
X_t\mid(X_{t-1},Y_{t-1}) \stackrel{D}= 0.5X_{t-1} + U \sqrt{25Y_{t-1}^2+1}
$$
where $U\stackrel{D}=\mathcal{N}(0,1)$. Hence, by the properties of the expectiles, we get 
\[
\mu_\tau(X_t\mid X_{t-1},Y_{t-1})
=
0.5X_{t-1}+k_\tau\sqrt{25Y_{t-1}^2+1},
\]
where $k_\tau = \mu_\tau(U)$. 
In the case of $\tau\neq 1/2$, we have $k_\tau \neq 0$, so that $k_\tau\sqrt{25Y_{t-1}^2+1}$ is a non-trivial random variable, and consequently, $\mu_\tau(X_t\mid X_{t-1},Y_{t-1})$ differs almost surely from $\mu_\tau(X_t\mid X_{t-1})$. Hence, in this case, there is individual Granger causality in the $\tau$-expectile from $Y\to X$. The same argument follows for the case of $Z\to  X$ given the symmetry of P3.\\\indent 
Lastly, we show the joint Granger causality in the $\tau$-expectile from $(Y,Z)$ to $X$ for every $\tau\in(0,1)$. Conditioning on $(X_{t-1},Y_{t-1},Z_{t-1})$, the only remaining randomness in $X_t$ comes from $\eta_{x,t}\stackrel{D}=\mathcal{N}(0,1)$, so we obtain 
$$
\mu_\tau(X_t\mid X_{t-1},Y_{t-1},Z_{t-1}) = 0.5X_{t-1} + 5Y_{t-1}Z_{t-1} + k_\tau.
$$
Since $Y_{t-1}Z_{t-1}$ is a non-trivial random variable, the above quantity  differs almost surely from $\mu_\tau(X_t\mid X_{t-1})$ for every $\tau$. Therefore, for all values of $\tau$, there is joint Granger causality in $\tau$-expectile from $(Y,Z)\to X$. 
\end{document}